\documentclass[letterpaper, 10 pt, conference]{ieeeconf}  %

\usepackage{algorithmic}
\usepackage{graphicx}
\usepackage{textcomp}

\usepackage{amsmath}
\usepackage{amssymb}

\usepackage{fancyhdr}

\usepackage{xcolor}
\colorlet{citeblue}{blue!50!black}
\colorlet{linkred}{red!50!black}

\makeatletter
\let\NAT@parse\undefined
\makeatother

\usepackage{hyperref}
\hypersetup{
  colorlinks,
  linkcolor=linkred,
  citecolor=citeblue,
  urlcolor=citeblue,
  breaklinks=true,   %
  linktoc=all,
}

\newcommand{\rev}[1]{{#1}}
\newcommand{\revv}[1]{{#1}}

\newcommand{\seqq}[1]{\underline{#1}}

\newtheorem{defn}{Definition}[section]
\newtheorem{proposition}[defn]{Proposition}
\newtheorem{lemma}[defn]{Lemma}

\newtheorem{remark}[defn]{Remark}
\newtheorem{theorem}[defn]{Theorem}

\newtheorem{assumption}[defn]{Assumption}

\IEEEoverridecommandlockouts                              %

\newcommand{\mytitle}{\textbf{Accepted version.} To appear in \emph{IEEE Control Systems Letters}.
\copyright 2023 IEEE. Personal use of this material is permitted. Permission
from IEEE must be obtained for all other uses, in any current or future
media, including reprinting/republishing this material for advertising or
promotional purposes, creating new collective works, for resale or
redistribution to servers or lists, or reuse of any copyrighted component of
this work in other works.}
\begin{document}
\newcommand{\Pbb}{\mathcal{P}} %
\newcommand{\dkm}{d_{\text{KR}}} %
\newcommand{\convmode}[1]{\overset{#1}{\longrightarrow}}
\newcommand{\Pone}{\mathcal{P}_1}
\newcommand{\Permutations}{\mathcal{S}}
\newcommand{\EmpMeas}[2]{\mathcal{E}_{#1}(#2)} %
\newcommand{\muhat}{\hat \mu} %
\newcommand{\KME}{\Pi}
\newcommand{\KMEk}{\Pi_k} %
\newcommand{\KMEhat}{\hat{\Pi}_k} 
\newcommand{\MMD}{\gamma}
\newcommand{\MMDk}{\gamma_k} %
\newcommand{\decoM}[1]{{{#1}^{[M]}}} %
\newcommand{\deco}[2]{{{{#1}^{[#2]}}}}
\newcommand{\seq}[1]{\underline{#1}} %
\newcommand{\setc}[1]{\mathbb{#1}} %

\newcommand{\supp}{\mathrm{supp}} %
\newcommand{\openNeighbourhood}{\mathcal{U}}
\newcommand{\closedBall}{\bar{\mathbf{B}}}
\newcommand{\openBall}{\mathbf{B}}

\newcommand{\cK}{\mathcal{K}}
\newcommand{\cKi}{\mathcal{K}_\infty}
\newcommand{\cKL}{\mathcal{KL}}

\newcommand{\cf}[1]{{\color{blue}#1}}

\newcommand{\N}{\mathbb{N}} %
\newcommand{\Z}{\mathbb{Z}} %
\newcommand{\R}{\mathbb{R}} %
\newcommand{\C}{\mathbb{C}} %
\newcommand{\K}{\mathbb{K}} %

\newcommand{\Rnn}{\mathbb{R}_{\geq 0}} %
\newcommand{\Rp}{\mathbb{R}_{>0}} %
\newcommand{\Rx}{\mathbb{R} \cup \{ \infty \}} %
\newcommand{\Np}{\mathbb{N}_+} %
\newcommand{\Npe}{\Np\cup\{\infty\}} %
\newcommand{\Ne}{\N\cup\{\infty\}} %

\title{Mean field limits for discrete-time dynamical systems via kernel mean embeddings}
\author{Christian Fiedler, Michael Herty, and Sebastian Trimpe
\thanks{
This work is funded in part under the Excellence Strategy of the Federal Government and the Länder (G:(DE-82)EXS-SF-SFDdM035), which the authors gratefully acknowledge.}
\thanks{C. Fiedler and S. Trimpe are with the Institute for Data Science in Mechanical Engineering (DSME), RWTH Aachen University, Aachen, Germany (e-mail: \{fiedler,trimpe\}@dsme.rwth-aachen.de)}
\thanks{M. Herty is with the Institute for Geometry and Practical Mathematics (IGPM), RWTH Aachen University, Aachen Germany (e-mail: herty@igpm.rwth-aachen.de)}
}

\maketitle
\thispagestyle{fancy} %

\begin{abstract}
Mean field limits are an important tool in the context of large\revv{-}scale dynamical systems, in particular, when studying multiagent and interacting particle systems. While the continuous-time theory is well-developed, few works have considered mean field limits for deterministic discrete-time systems, which are relevant for the analysis and control of large\revv{-}scale discrete\revv{-}time multiagent system.
We prove existence results for the mean field limit of very general discrete-time control systems, for which we utilize kernel mean embeddings. These results are then applied in a typical optimal control setup, where we establish the mean field limit of the relaxed dynamic programming principle.
Our results can serve as a rigorous foundation for many applications of mean field approaches for discrete-time dynamical systems.
\end{abstract}

\section{Introduction}
\label{sec:introduction}
Multiagent systems (MAS) or interacting particle systems (IPS) play an important role in various fields of science and engineering \cite{bellomo2017active,bullo2020lectures}. 
In many interesting applications, \rev{like \revv{controlling} large swarms of robots \cite{cui2023scalable}}, the number of agents or particles becomes very large, which poses challenges for the modeling, simulation and control of such systems.
A frequently used approach is to go from the microscopic level -- modeling each agent -- to the mesoscopic level, modeling only the distribution of the agents over the state space \cite{pareschi2013interacting,naldi2010mathematical}.
One way to do this is the mean field limit (MFL), which corresponds formally to the limit of infinitely many agents.
It has been very successfully employed for IPS modeled by ordinary differential equations (ODEs), leading to partial differential equations (PDEs), usually of Vlasov-Fokker-Planck type \rev{\cite{golse2016dynamics,pareschi2013interacting}}.
Recently, also the MFL of controlled systems has been considered, see for example \cite{herty2015mean,fornasier2019mean}, resulting in mean field (optimal) control.
However, the vast majority of past work has focused on continuous-time systems modeled by ODEs and the corresponding mean field PDEs, instead of discrete-time MAS like \cite{cucker2007emergent,yin2021convergence}, \revv{which are particularly relevant for engineering applications \cite{bullo2020lectures}}.
Here, we consider the latter type of systems and their MFL, as well as corresponding control problems.

\revv{First, we tackle the challenging problem of existence of mean field limits.
In contrast to the continuous-time case, where a certain PDE is connected with the microscopic ODEs, in the discrete-time case we are concerned with the MFL of the transition functions of the microscopic dynamics.}
For scalar-valued functions, a well-known existence result from the literature on mean field games is available (cf. Proposition \ref{prop:mflFuncs} below), which can be used for the MFL of stage cost functions in optimal control \cite{herty2017performance}.
However, no comparable result is available for transition functions, i.e., functions with a growing number of inputs and outputs.
Our first main contributions (in Section \ref{sec:mfl}) are existence results for MFLs of functions arising in the MFL of discrete-time dynamical systems. 
To the best of our knowledge, these are the first such existence results suitable for general discrete-time control systems, and they are also of independent interest due to their generality.
In order to prove these results, we use kernel mean embeddings (KMEs) which map probability distributions into reproducing kernel Hilbert spaces (RKHSs), cf. Section \ref{sec:prelims}, which collects all necessary preliminaries. 
\rev{Introducing KMEs as a tool in this context, is another major contribution of this work.}
\rev{In Section \ref{sec:dtds}, we apply these results to discrete-time control systems describing MAS, considering also stage costs with a view towards optimal control applications. As our final contribution, we prove a MFL result for relaxed dynamic programming, as used in the theory of nonlinear model predictive control (NMPC) \cite{grune2017nonlinear}.}

\revv{Our results provide a solid theoretical foundation for mean field approaches for the control and simulation of large-scale discrete-time MAS, which are increasingly common in application areas like cooperative robotics, power systems or social dynamics.}

\section{Preliminaries} \label{sec:prelims}
\emph{Sets}
For a set $X\not=\emptyset$ and a tuple $\vec x \in X^M$, $M\in\Np$, we denote by $x_i$ the $i$-th component of $\vec x$.
Denote by $\Permutations_M$ the set of all permutations of $\{1,\ldots,M\}$.
Let now $\vec x \in X^M$. 
Given $\sigma\in\Permutations_M$, we define $\sigma \vec x = (x_{\sigma(1)}, \ldots, x_{\sigma(M)})$.
Furthermore, we define the empirical measure with $M$ atoms $x_1,\ldots,x_M$ as $\muhat[\vec x] = \frac{1}{M} \sum_{m=1}^M \delta_{x_m}$,
where $\delta_x$ is the Dirac measure at atom $x\in X$,
and denote by $\EmpMeas{M}{X}$
the set of all empirical measures with $M$ atoms.

\emph{RKHS\rev{s}}
A kernel on $X$ is a function $k:X \times X \rightarrow \R$ such that $(k(x_i,x_j))_{i,j}$ is symmetric positive semidefinite for all $\vec x \in X^M$ and all $M\in\Np$.
Given a Hilbert space $(H,\langle\cdot,\cdot\rangle_H)$ consisting of real-valued functions on $X$, we say that $k$ is a reproducing kernel for $H$ if $k(\cdot,x)\in H$ for all $x\in X$, and $f(x)=\langle f, k(\cdot,x)\rangle_H$ for all $f\in H$ and $x\in X$. 
If $H$ has a reproducing kernel, we call $H$ a reproducing kernel Hilbert space (RKHS). 
It is well-known that for every kernel there exists a unique RKHS that has the given kernel as its reproducing kernel, and we denote this RKHS by $(H_k,\langle\cdot,\cdot,\rangle_k)$.
If $k$ is a kernel on $X$, we define
\begin{equation*}
    d_k: X \times X \rightarrow \Rnn,\quad d_k(x,x') = \|k(\cdot,x) - k(\cdot,x')\|_k,
\end{equation*}
which is a (semi)metric on $X$, called the kernel (semi)metric (associated with $k$).
For more details \rev{and proofs of all of these facts}, see \cite[Chp~4]{steinwart2008support}.

\emph{Probability measures} 
Let now $(X,d_X)$ be a metric space and denote by $\Pbb(X)$ the set of Borel probability measures on $X$, which we endow with the topology of weak convergence of probability measures. 
Recall that for $\mu_n,\mu\in\Pbb(X)$, we say that $\mu_n \rightarrow \mu$ weakly if for all bounded and continuous $\phi:X\rightarrow\R$, we have $\lim_{n\rightarrow\infty} \int_X \phi(x)\mathrm{d}\mu_n(x) \rightarrow \int_X \phi(x)\mathrm{d}\mu(x)$.
The topology of weak convergence can be metrized by the Kantorowich-Rubinstein metric $\dkm$. %
Note that for $X$ separable $\dkm$ corresponds to the 1-Wasserstein distance, and if $X$ is compact, then $(\Pbb(X),\dkm)$ is also compact, cf. \cite[Chp~11]{dudley2018real}.

\emph{KME and MMD} 
Let $k:X\times X \rightarrow \R$ be a Borel-measurable kernel. Furthermore, assume that it is bounded, i.e., $\sup_{x,x'\in X} |k(x,x')| < \infty$.
In this case, $X \ni x \mapsto k(\cdot,x)\in H_k$ is Bochner-integrable w.r.t. every $\mu\in\Pbb(X)$, and we define
\begin{equation}
    \KMEk: \Pbb(X) \rightarrow H_k, \quad \mu \mapsto \int_X k(\cdot,x)\mathrm{d}\mu.
\end{equation}
For $\mu\in\Pbb(X)$, we call $\KMEk(\mu)$ the kernel mean embedding (KME) of $\mu$ in $H_k$.\footnote{This terminology is explained by the fact that in the present setting, for all $\mu\in\Pbb(X)$, any $f\in H_K$ is $\mu$-integrable, and $\int_X f(x)\mathrm{d}\mu=\langle f, \KMEk(\mu)\rangle_k$.}
If the map $\KMEk$ is injective, then we call $k$ characteristic, \rev{cf. \cite{sriperumbudur2011universality} for many examples and conditions for this property}.
To simplify notation, define additionally $\KMEhat(\vec x) = \KMEk(\muhat[\vec x])$ for $\vec x \in X^M$.
For convenience, we record the following simple fact.
\begin{lemma}
The set $\KMEk\left(\Pbb(X)\right)$ is convex.
\end{lemma}
\begin{proof}
\revv{Since $\KMEk$ is a linear map and $\Pbb(X)$ is a convex set, also $\KMEk\left(\Pbb(X)\right)$ is convex.}
\QED
\end{proof}
Finally, define the maximum mean discrepancy (MMD) w.r.t. $k$ by
\begin{equation*}
    \MMDk: \Pbb(X)\times\Pbb(X) \rightarrow \Rnn, \quad
    \MMDk(\mu,\nu) = \| \KMEk(\mu)-\KMEk(\nu) \|_k.
\end{equation*}
$\MMDk$ is a semimetric on $\Pbb(X)$, and if $k$ is characteristic, $\MMDk$ is a metric.
\begin{lemma}
    If $(X,d_k)$ is a compact metric space, then $\KMEk\left(\Pbb(X)\right)$ is compact.
\end{lemma}
\begin{proof}
Recall that $(\Pbb(X),\dkm)$ is compact, where $\dkm$ is defined with $d_X=d_k$. Next, since $(X,d_k)$ is compact, it is separable, and hence \cite[Theorem~21]{sriperumbudur2010hilbert} implies that for all $\mu,\nu\in\Pbb(X)$
\begin{equation*}
    \MMDk(\mu,\nu)=\|\KMEk(\mu)-\KMEk(\nu)\|_k \leq \dkm(\mu,\nu),
\end{equation*}
which shows that $\KMEk: (\Pbb,\dkm)\rightarrow (H_k,\|\cdot\|_k)$ is 1-Lipschitz continuous.
Altogether, $\KMEk(\Pbb(X))$ is the image of a compact set under a continuous map, so it is compact. \QED
\end{proof}
For additional background on KMEs and MMD, see \cite{sriperumbudur2010hilbert}.

\emph{Mean field limit of functions} 
Given functions $f_M: X^M\rightarrow\R$, $M\in\Np$, and $f:\Pbb(X)\rightarrow\R$\rev{,} on a metric space $(X,d_X)$, we say that $f$ is the mean field limit of $(f_M)_M$ if $\sup_{\vec x \in X^M} |f_M(\vec x)-f(\muhat[\vec x])| \rightarrow 0$ for $M\rightarrow \infty$, and we denote this by $f_M \convmode{\Pone} f$.
The following result is well-known, cf. \cite[Lemma~1.2]{carmona2018probabilistic}.
\begin{proposition} \label{prop:mflFuncs}
Let $(X,d_X)$ be a compact metric space, and assume that 
1) $\forall M\in\Np, \vec x\in X^M, \sigma\in\Permutations_M$: $f_M(\sigma \vec x)=f_M(\vec x)$;
2) $\exists B\in\Rnn \forall M\in \Np \forall \vec x \in X^M$: $|f_M(\vec x)|\leq B$;
3) $\exists L\in\Rnn \forall M\in \Np \forall \vec x, \vec x'\in X^M$: $|f_M(\vec x)-f_M(\vec x')|\leq L\dkm(\muhat[\vec x], \muhat[\vec x'])$.
Then there exists a subsequence $(f_{M_\ell})_\ell$ and a $\dkm$-continuous $f:\Pbb(X)\rightarrow\R$ with $f_{M_\ell}\convmode{\Pone} f$ for $\ell\rightarrow\infty$.
\end{proposition}
The proof of this result relies on the McShane extension, which %
requires the extended function to be real-valued.
This poses a major obstacle to obtaining an analogous result for transition functions.
To circumvent this problem, we use Kirszbraun's theorem, cf. \cite[Theorem~4.2.3]{cobzacs2019lipschitz}, which in turn requires a Hilbert space setting.
This motivates our usage of KMEs, which allows  us to prove the results of the next section.
\section{New Mean Field Limit Existence Results} \label{sec:mfl}
We now present and prove new mean field limit existence results that are tailored to discrete-time control systems. 
\subsection{Systems with control input}
Our first main result concerns functions of the form of transition functions of discrete-time dynamics with input.
We will apply this result to such systems in Section \ref{sec:dtds}.
\begin{theorem} \label{thm:mflWithInput}
Let $(X,d_X)$ be a metric space, $H$ a Hilbert space and $U\subseteq H$ compact, $k$ a measurable, bounded and characteristic kernel such that $\KMEk(\Pbb(X))$ is compact. Consider a sequence of functions $f_M: X^M \times U \rightarrow X^M$, $M\in\Np$, such that
1) $ \forall M\in\Np \: \forall \vec x \in X^M, \sigma\in\Permutations_M, u \in U$:     $\muhat[f_M(\sigma \vec x, u)]= \muhat[f_M(\vec x, u)]$;
2) $\exists L\in\Rnn \: \forall M\in\Np \forall \vec x,x'\in X^M$, $u, u'\in U$:
    \begin{align}
        & \MMDk(\muhat[f_M(\vec x, u)], \muhat[f_M(\vec x', u')]) \nonumber \\
        & \hspace{0.5cm} \leq L\left\|(\KMEk(\muhat[\vec x]),u) - (\KMEk(\muhat[\vec x']),u')\right\|_{H_k\times H},
    \end{align}
    where $\|\cdot\|_{H_k\times H}$ is the norm for the product of the Hilbert spaces $H_k$ and $H$.

Then there exist a subsequence $(f_{M_\ell})_\ell$ and an $L$-Lipschitz continuous function $F:\KMEk\left(\Pbb(X)\right) \times U \rightarrow\KMEk\left(\Pbb(X)\right)$ such that
\begin{equation}
     \lim_{\ell\rightarrow\infty} \sup_{\vec x \in X^{M_\ell}, u\in U} \|\KMEhat\left(f_{M_\ell}(\vec x, u)\right)-F(\KMEhat(\vec x), u)\|_k = 0
\end{equation}
and
\begin{equation} \label{eq:mfl:f}
     \lim_{\ell\rightarrow\infty} \sup_{\vec x \in X^{M_\ell}, u \in U} \MMDk(\muhat[f_{M_\ell}(\vec x, u)], f(\muhat[\vec x],u)) = 0,
\end{equation}
where we defined $f(\mu,u)=\KMEk^{-1}(F(\KMEk(\mu),u))$.
\end{theorem}
\rev{
\begin{remark}
The convergence notion \eqref{eq:mfl:f} is a direct adaption of the established concept of mean field convergence of functions as used in Proposition \ref{prop:mflFuncs}. The new notion \eqref{eq:mfl:f} is tailored to the setting of discrete-time control systems, as illustrated by the results in Section \ref{sec:dtds}.
\end{remark} }

\begin{proof}
Due to property 1) of the $f_M$, the mappings $\tilde f_M: \EmpMeas{M}{X} \times U \rightarrow \EmpMeas{M}{X}$ given by
\begin{equation*}
    \tilde f_M \left( \frac{1}{M} \sum_{m=1}^M \delta_{x_m}, u \right) = \muhat[f_M(x_1,\ldots,x_M, u)]
\end{equation*}
are well-defined for all $M\in\Np$. Furthermore, since $k$ is characteristic, the mappings
\begin{equation*}
    \tilde F_M = \KMEk \circ \tilde f_M \circ \KMEk\lvert_{\KMEk(\EmpMeas{M}{X})}^{-1}
\end{equation*}
are well-defined. Let now $M\in\Np$, $g,g'\in\KMEk(\EmpMeas{M}{X})$, $u,u'\in U$ be arbitrary, and choose $\vec x, \vec x'\in X^M$ such that
\begin{equation*}
    g = \KMEk\left(\frac{1}{M} \sum_{m=1}^M \delta_{x_m} \right), \quad
    g'=  \KMEk\left(\frac{1}{M} \sum_{m=1}^M \delta_{x'_m} \right).
\end{equation*}
We then have
\begin{align*}
    & \| \tilde F_M(g,u) - \tilde F_M(g',u') \|_k \\
    & \hspace{0.5cm} = \left\|\KMEhat\left(f_M(\vec x, u)\right) - \KMEhat\left(f_M(\vec x', u')\right)\right\|_k \\
    & \hspace{0.5cm}=   \MMDk(\muhat[f_M(\vec x, u)], \muhat[f_M(\vec x', u')]) \\
    & \hspace{0.5cm}\leq L\left\|(\KMEk(\muhat[\vec x]),u) - (\KMEk(\muhat[\vec x']),u')\right\|_{H_k\times H} \\
    & \hspace{0.5cm}= L\left\|(g,u) - (g',u')\right\|_{H_k\times H},
\end{align*}
which shows that all $\tilde F_M$ are $L$-Lipschitz continuous. 
Since $H_k\times H$ is a Hilbert space, Kirszbraun's theorem ensures that there exist $L$-Lipschitz continuous mappings $\bar F_M: H_k \times H \rightarrow H_k$ with $\bar F_M\lvert_{\KMEk(\Pbb(X))\times U} = \tilde F_M$, for all $M\in\Np$.
Recall that $\KMEk(\Pbb(X))$ is convex and by assumption also compact, so there exists the orthogonal projection $P_{\KMEk(\Pbb(X))}$ from $H_k$ onto $\KMEk(\Pbb(X))$. Define now for all $M\in\Np$ the mappings $F_M: \KMEk(\Pbb(X))\rightarrow \KMEk(\Pbb(X))$ by
\begin{equation*}
    F_M = P_{\KMEk(\Pbb(X))} \circ \bar F_M \lvert_{\KMEk(\Pbb(X)) \times U}.
\end{equation*}
We have for all $M\in\Np$, $g,g'\in\KMEk(\EmpMeas{M}{ X})$, $u,u'\in U$ that
\begin{align*}
    & \| F_M(g,u) - F_M(g',u')\|_k \\
    & \hspace{0.5cm} = \| P_{\KMEk(\Pbb(X))}\left(\bar F_M(g,u) - \bar F_M(g',u') \right) \|_k \\
    & \hspace{0.5cm}\leq \|  P_{\KMEk(\Pbb(X))} \|_{L(H_k)} \| \bar F_M(g,u) - \bar F_M(g',u') \|_k \\
    & \hspace{0.5cm}\leq  L\left\|(g,u) - (g',u')\right\|_{H_k\times H},
\end{align*}
where we used that $ \|P_{\KMEk(\Pbb(X))} \|_{L(H_k)} \leq 1$ (since $P_{\KMEk(\Pbb(X))}$ is a projection). This shows
that all $F_M$ are $L$-Lipschitz continuous. 
Since $\KMEk(\Pbb(X)) \times U$ and $\KMEk(\Pbb(X))$ are compact, we now have a sequence $(F_M)_{M\in\Np}$ of equicontinuous functions defined on a compact input set, with their range contained in a compact set, so the Arzela-Ascoli theorem asserts that there exist a subsequence $(F_{M_\ell})_{\ell}$ and an $L$-Lipschitz continuous function $F: \KMEk(\Pbb(X)) \rightarrow \KMEk(\Pbb(X))$ with
\begin{equation*}
    \lim_{\ell\rightarrow\infty} \sup_{g \in \KMEk(\Pbb(X)), u \in U} \| F_{M_\ell}(g,u) - F(g,u)\|_k \revv{=0},
\end{equation*}
which implies that
\begin{equation*}
 \lim_{\ell\rightarrow\infty} \sup_{\vec x \in X^{M_\ell}, u \in U} \| F_{M_\ell}(\revv{\KMEhat(\vec x)},u) - F(\revv{\KMEhat(\vec x)},u)\|_k \revv{=0}.
\end{equation*}
Observe that for all $M\in\Np$, $\vec x \in X^M$ and $u\in U$ we have
\begin{align*}
    F_M(\KMEk(\muhat[\vec x]), u) & = P_{\KMEk(\Pbb(X))}\left( \bar F_M(\KMEk(\muhat[\vec x]), u) \right) \\
    & = P_{\KMEk(\Pbb(X))}\left( \tilde F_M(\KMEk(\muhat[\vec x]), u) \right) \\
    & = P_{\KMEk(\Pbb(X))}\left( \KMEk(\muhat[f_M(\vec x, u)]) \right) \\
    & = \KMEk(\muhat[f_M(\vec x, u)]),
\end{align*}
which implies
\begin{equation*}
 \lim_{\ell\rightarrow\infty} \sup_{\vec x \in X^{M_\ell}, u \in U} \| \KMEhat(f_{M_\ell}(\vec x, u)) - F(\KMEhat(\vec x),u)\|_k =0,
\end{equation*}
and since $\KMEk$ is injective, we can set $f(\mu,u)=\KMEk^{-1}(F(\KMEk(\mu),u))$ and get \eqref{eq:mfl:f}.
\QED
\end{proof}

\subsection{Feedback maps}
Our second main result is concerned with functions that can be used as feedback maps for discrete-time control systems, cf. Theorem \ref{thm:mflRDP} for an example of such a situation.
\begin{theorem} \label{thm:mflFbMap}
Let $(X,d_X)$ be a metric space, $k:X\times X \rightarrow \R$ a bounded, Borel-measurable and characteristic kernel on $X$ such that $\KMEk(\Pbb(X))$ is compact, $H$ a Hilbert space and $C\subseteq H$ a compact and convex subset. Consider maps $g_M: X^M \rightarrow C$, $M\in\Np$, such that
1) $ \forall M\in\Np,\,\vec x \in X^M, \,\sigma\in\Permutations_M, \,u \in U$: $g_M(\sigma \vec x)= g_M(\vec x)$;
2) $\exists L\in\Rnn \, \forall M\in\Np, \,\vec x,x'\in X^M$:
    \begin{equation}
        \|g_M(\vec x) - g_M(\vec x') \|_H \leq L \MMDk(\muhat[\vec x], \muhat[\vec x']).
    \end{equation}

Then there exist a subsequence $(g_{M_\ell})_\ell$ and an $L$-Lipschitz continuous map $G: \KMEk(\Pbb(X)) \rightarrow C$ such that
\begin{equation}
    \lim_{\ell\rightarrow\infty} \sup_{\vec x \in X^{M_\ell}} \| g_{M_\ell}(\vec x) - G(\KMEk(\muhat[\vec x]))\|_H \rev{=0}
\end{equation}
and
\begin{equation} \label{eq:mfl:g}
    \lim_{\ell\rightarrow\infty} \sup_{\vec x \in X^{M_\ell}} \| g_{M_\ell}(\vec x) - g(\muhat[\vec x])\|_H \rev{=0},
\end{equation}
where we defined $g=G \circ \KMEk$.
\end{theorem}
\begin{remark}
$g$ is also $L$-Lipschitz continuous as a map on $(\Pbb(X),\MMDk)$.
\end{remark}
\begin{proof}
The proof follows a similar strategy as used in the proof of Theorem \ref{thm:mflWithInput}.
Since all $g_M$ are permutation invariant, we can define the maps $ \tilde g_M: \EmpMeas{M}{X} \rightarrow C$ by
\begin{equation*}
    \tilde g_M\left(\frac{1}{M} \sum_{m=1}^M \delta_{x_m} \right)= g_M(x_1,\ldots,x_M),
\end{equation*}
and since $k$ is characteristic, we can further define $\tilde G_M = \tilde g_M \circ \KMEk^{-1}\lvert_{\KMEk(\EmpMeas{M}{X})}$.
Observe that for all $M\in\Np$ and $\vec x \in X^M$ we have by construction $\tilde G_M(\KMEk(\muhat[\vec x]))=g_M(\vec x)$.
Furthermore, for all $M\in\Np$ and $\vec x, \vec x' \in X^M$ we have
\begin{align*}
    \|\tilde G_M(\KMEhat(\vec x)) - \tilde G_M(\KMEhat(\vec x'))\|_H & = \|g_M(\vec x) - g_M(\vec x') \|_H \\
    &  \leq L\MMDk(\muhat[\vec x], \muhat[\vec x']), 
\end{align*}
so Kirszbraun's theorem ensures the existence of $L$-Lipschitz continuous maps $\bar G_M: H_k \rightarrow H$ with $\bar G_M\lvert_{\KMEk(\EmpMeas{M}{X})} = g_M$ for all $M\in\Np$.
Since $C$ is compact and convex by assumption, the orthogonal projection $P_C$ exists, and we can define $G_M=P_C \circ \bar G_M\lvert_{\KMEk(\Pbb(X))}$.
Since $\|P_C\|_{L(H)}\leq 1$, we have for all $M\in\Np$, $f_1,f_2\in \KMEk(\Pbb(X))$ that
\begin{align*}
    \|G_M(f_1)-G_M(f_2)\|_H & = \|P_C(\bar G_M(f_1) - \bar G_M(f_2))\|_H \\
    & \hspace{-0.5cm} \leq \|P_C\|_{L(H)} \| \bar G_M(f_1) - \bar G_M(f_2)\|_H \\
    & \hspace{-0.5cm}\leq L\|f_1-f_2\|_k.
\end{align*}
The sequence $(G_M)_M$ is therefore equicontinuous, defined on a compact set (since by assumption $\KMEk(\Pbb(X))$ is compact) with a compact codomain, so by the Arzela-Ascoli theorem there exist a subsequence $(G_{M_\ell})_\ell$ and an $L$-Lipschitz continuous map $G: \KMEk(\Pbb(X))\rightarrow C$ such that
\begin{equation*}
    \lim_{\ell\rightarrow\infty} \sup_{f \in \KMEk(\Pbb(X))} \| G_{M_\ell}(f) - G(f)\|_H = 0,
\end{equation*}
which implies
\begin{equation*}
    \lim_{\ell\rightarrow\infty} \sup_{\vec x \in X^{M_\ell}} \|g_{M_\ell}(\vec x)  - G(\KMEk(\muhat[\vec x]))\|_H = 0.
\end{equation*}
Finally, since $k$ is characteristic, we can \revv{set} \rev{$g=G \circ \KMEk$}, and by definition we then have \eqref{eq:mfl:g}.
\QED
\end{proof}

\section{Application to Discrete-Time Systems} \label{sec:dtds}
We now apply our existence results to discrete-time control systems, \rev{in particular, large\revv{-}scale MAS.}
First, we specify an appropriate setup that ensures the existence of the mean field dynamics \rev{and mean field stage cost.}
Next, we show that in our setting also the corresponding total cost functional has a mean field limit.
Finally, we prove a result on relaxed dynamic programming in the mean field limit.
\subsection{Setup}
Let $X\not=\emptyset$ be some set, $k: X\times X \rightarrow \R$ a kernel on $X$ and denote by $d_k$ the corresponding kernel metric. From now on, we make the following assumption.
\begin{assumption} \label{assump:space}
The (semi)metric space $(X,d_k)$ is compact, and the kernel $k$ is bounded, Borel-measurable (w.r.t. the Borel $\sigma$-algebra on $(X,d_k)$) and characteristic.
\end{assumption}
Recall that under these assumptions $\KMEk(\Pbb(X))$ is compact.
Consider now a sequence of discrete-time dynamical systems
\begin{equation}
    \decoM{\vec x_+} = \decoM{f}(\decoM{\vec x},u), \: M\in\Np
\end{equation}
with transition functions $\decoM{f}: X^M \times U \rightarrow X^M$, where $X^M$ is the state space and $U$ the input space. Given an initial state $\decoM{\vec x_0}\in X^M$ and a control input sequence $\seq{u}\in U^N$, a state-trajectory $\decoM{\vec x}(\cdot;\decoM{\vec x_0},\seq{u})$ is induced by
\begin{align*}
    \decoM{\vec x}(0;\decoM{\vec x_0},\seq{u}) & = \decoM{\vec x_0} \\
    \decoM{\vec x}(n+1;\vec x_0^{[M]},\seq{u}) & = \decoM{f}(\decoM{\vec x}(n;\decoM{\vec x_0},\seq{u}), \revv{\seqq{u}}(n)),
\end{align*}
where $n=0,\ldots,N-1$.
We make the following assumption on these systems.
\begin{assumption} \label{assump:sys}
1) $U\subseteq H$ is compact, where $H$ is a Hilbert space;
2) all $f_M$ are permutation equivariant in the state, i.e., $\forall M\in\Np, \vec x \in X^M, \sigma\in\Permutations_M, u \in U:
    \: \decoM{f}(\sigma \vec x, u)= \sigma\decoM{f}(\vec x, u)$;
3) the $\decoM{f}$ are uniformly Lipschitz-continuous, i.e., $\exists L_f\in\Rnn \: \forall M\in\Np,  \vec x,x'\in X^M$, $u, u'\in U$:
    \begin{align*}
        & \MMDk(\muhat[\decoM{f}(\vec x, u)], \muhat[\decoM{f}(\vec x', u')]) \\
        & \hspace{0.5cm} \leq L_f\left\|(\KMEhat(\vec x),u) - (\KMEhat(\vec x'),u')\right\|_{H_k\times H}.
    \end{align*}
\end{assumption}
\rev{
\begin{remark}
    Permutation equi-variance is fulfilled by essentially all common discrete-time MAS, \revv{in particular, as appearing in engineering}, cf. \cite{bullo2020lectures}.
\end{remark} }
Assumptions \ref{assump:space} and \ref{assump:sys} together allow to apply Theorem \ref{thm:mflWithInput}, so there exist a subsequence $(\deco{f}{M_m})_m$ and a map $F: \KMEk(\Pbb(X))\rightarrow\KMEk(\Pbb(X))$ such that 
\begin{equation*}
      \lim_{\ell\rightarrow\infty} \sup_{\substack{\vec x \in X^{M_m}\\ u\in U}} \|\KMEhat\left(\deco{f}{M_m}(\vec x, u)\right)-F(\KMEhat(\vec x), u)\|_k = 0.
\end{equation*}
This map is also $L_f$-Lipschitz continuous.
Since $k$ is characteristic, we can define the function $f=\KMEk^{-1}\circ F \circ \KMEk: \Pbb(X)\times U \rightarrow \Pbb(X)$, so
\begin{equation*}
      \lim_{m\rightarrow\infty} \sup_{\vec x \in X^{M_m}, u \in U} \MMDk(\muhat[\deco{f}{M_\ell}(\vec x, u)], f(\muhat[\vec x],u)) = 0.
\end{equation*}
The map $f$ induces another discrete-time dynamical system
\begin{equation}
    \mu_+ = f(\mu,u),
\end{equation}
where $\Pbb(X)$ is the state space and $U$ the input space. A given initial state $\mu_0\in\Pbb(X)$ and control sequence $\seq{u}\in U^N$ induce a state trajectory $\mu(\cdot;\mu_0,\seq{u})$ by
\begin{align*}
    \mu(0;\mu_0,\seq{u}) & = \mu_0 \\
     \mu(n+1;\mu_0,\seq{u}) & = f(\mu(n;\mu_0,\seq{u}), \revv{\seqq{u}}(n)) \: \forall n=0,\ldots,N-1.
\end{align*}

\rev{Motivated by optimal control applications,} consider a sequence of stage cost functions $\decoM{\ell}: X^M \times U \rightarrow \R$, and the associated finite-horizon total cost functionals $J_N^{[M]}: X^M \times U^N \rightarrow \R$, $N\in\Np$, defined by
\begin{equation}
    J_N^{[M]}(\vec x_0, \seq{u}) = \sum_{n=0}^{N-1} \decoM{\ell}(\decoM{\vec x}(n;\vec x_0, \seq{u}), \revv{\seqq{u}}(n)).
\end{equation}
We make the following assumption on the stage cost functions $\decoM{\ell}$.
\begin{assumption} \label{assump:stagecost}
1) All $\decoM{\ell}$ are permutation-invariant in the state variable, i.e., $\forall M\in\Np, \vec x \in X^M, \sigma\in\Permutations_M, u \in U$:    $\decoM{\ell}(\sigma \vec x, u)= \decoM{\ell}(\vec x, u)$;
2) the $\decoM{\ell}$ are uniformly bounded, i.e., $\exists B_\ell\in\Rnn \: \forall M\in\Np, \vec x\in X^M, u\in U:$ $|\decoM{\ell}(\vec x, u)|\leq B_\ell$;
3) the stage cost functions are uniformly Lipschitz-continuous, $\exists L_\ell\in\Rnn \forall M\in\Np, \vec x,x' \in X^M, u,u'\in U$
    \begin{equation*}
        |\decoM{\ell}(\vec x, u) - \decoM{\ell}(\vec x', u')| \leq L_\ell(\MMDk(\muhat[\vec x], \muhat[\vec x']) + \|u-u'\|_H).
    \end{equation*}
\end{assumption}
An inspection of the proof of Proposition \ref{prop:mflFuncs} shows that under Assumptions \ref{assump:space} and \ref{assump:stagecost} there exist a subsequence $(\deco{\ell}{M_{m_p}})_p$ and a function $\ell:\Pbb(X)\times U \rightarrow \R$ such that
\begin{equation} \label{eq:mflell}
    \lim_{p\rightarrow\infty} \sup_{\vec x \in X^{M_{m_p}}, u\in U} |\deco{\ell}{M_{m_p}}(\vec x, u) - \ell(\muhat[\vec x],u)| = 0.
\end{equation}
This can be used to define a corresponding total cost functional $J_N: \Pbb(X)\times U \rightarrow \R$, $N\in\Np$, by
\begin{equation}
    J_N(\mu_0, \seq{u}) = \sum_{n=0}^{N-1} \ell(\mu(n;\mu_0, \seq{u}), \revv{\seqq{u}}(n)).
\end{equation}

We now switch to the subsequence $(M_{m_p})_p$ and reindex by $M$ for readability, so that we can write
\begin{align}
    \lim_{M\rightarrow\infty} \sup_{\vec x \in X^M, u \in U} \MMDk(\muhat[\decoM{f}(\vec x, u)], f(\muhat[\vec x],u)) = 0 \\
    \lim_{M\rightarrow\infty} \sup_{\vec x \in X^M, u\in U} |\decoM{\ell}(\vec x, u) - \ell(\muhat[\vec x],u)| = 0.
\end{align}

\subsection{Mean field limit of $J_N^{[M]}$}
Next, we show that $J_N$ is indeed the MFL of $J_N^{[M]}$.
\begin{proposition} \label{prop:mflJ}
For all $N\in\Np$ we have
\begin{equation}
    \lim_{M\rightarrow\infty} \sup_{\substack{\vec x_0 \in X^M\\ \seq{u}\in U^N}} | J_N^{[M]}(\vec x_0, \seq{u}) - J_N(\muhat[\vec x_0],\seq{u})| = 0.
\end{equation}
\end{proposition}
For the proof we need a technical lemma.
\begin{lemma} \label{lem:trajectoryBound}
For all $M\in\Np$, $\vec x_0 \in X^M$, $N\in \Np$ and $\seq{u}\in U^N$ we have
\begin{align*}
    \MMDk(\muhat(N), \mu(N)
    & \leq
    \sum_{n=1}^N L_f^{n-1}\|\KMEk(\muhat(N-n+1)) \\
    & \hspace{1cm} - F(\muhat(N-n), \revv{\seqq{u}}(N-n))\|_k,
\end{align*} 
where we defined for brevity $\muhat(n)=\muhat[\decoM{\vec x}(n; \vec x_0, \seq{u})]$ and $\mu(n)=\mu(\rev{n};\muhat[\vec x],\seq{u})$.
\end{lemma}
This result can be shown using a standard induction argument, and hence the proof is omitted.

\emph{Proof (of Proposition \ref{prop:mflJ}):}
Let $N\in\Np$, $M\in\Np$, $\vec x_0\in X^M$ and $\seq{u}\in U^N$ be arbitrary, and define $\vec x(n)=\decoM{\vec x}(n;\vec x_0, \seq{u})$ and $\muhat(n)=\muhat[\vec x(n)]$, then we have
\begin{align*}
    & | J_N^{[M]}(\vec x_0, \seq{u}) - J_N(\muhat[\vec x_0],\seq{u})| \\
    & \leq \sum_{n=0}^{N-1} |\decoM{\ell}(\vec x(n),\revv{\seqq{u}}(n)) -  \ell(\mu(n;\muhat[\vec x_0], \seq{u}), \revv{\seqq{u}}(n))| \\
    & \leq \sum_{n=0}^{N-1} |\decoM{\ell}(\vec x(n),\revv{\seqq{u}}(n)) - \ell(\muhat(n), \revv{\seqq{u}}(n))| \\
        & \hspace{1cm} + |\ell(\muhat(n), \revv{\seqq{u}}(n)) -  \ell(\mu(n;\muhat[\vec x_0], \seq{u}), \revv{\seqq{u}}(n))|.
\end{align*}
Since $\ell$ is $L_\ell$-Lipschitz continuous, we have for all $n=0,\ldots,N-1$ that
\begin{align*}
     & |\ell(\muhat(n), \revv{\seqq{u}}(n)) -  \ell(\mu(n;\muhat[\vec x_0], \seq{u}), \revv{\seqq{u}}(n))| \\
     & \hspace{0.5cm} \leq L_\ell \MMDk(\muhat[\decoM{\vec x}(n;\vec x_0, \seq{u})], \mu(n;\muhat[\vec x_0], \seq{u})) \\
     & \hspace{0.5cm} \leq L_\ell  \sum_{i=1}^n L_f^{i-1}\|\KMEhat(\vec x(n-i+1)) \\
     & \hspace{1cm} - F(\muhat(n-i), \revv{\seqq{u}}(n-i))\|_k,
\end{align*}
where we used Lemma \ref{lem:trajectoryBound} in the second inequality.

Combining these bounds results in
\begin{align*}
    & \sup_{\substack{\vec x_0 \in X^M\\ \seq{u}\in U^N}} | J_N^{[M]}(\vec x_0, \seq{u}) - J_N(\muhat[\vec x_0],\seq{u})| \\
    & \leq  \sup_{\substack{\vec x_0 \in X^M\\ \seq{u}\in U^N}} \sum_{n=0}^{N-1} |\decoM{\ell}(\vec x(n),\revv{\seqq{u}}(n)) - \ell(\muhat(n), \revv{\seqq{u}}(n))| \\
        & \hspace{0.5cm} + \sup_{\substack{\vec x_0 \in X^M\\ \seq{u}\in U^N}} \sum_{n=0}^{N-1}  L_\ell  \sum_{i=1}^n  L_f^{i-1}\|\KMEk(\muhat(n-i+1)) \\
        & \hspace{1cm} - F(\muhat(n-i), \revv{\seqq{u}}(n-i))\|_k, \\
    & \leq \sum_{n=0}^{N-1}  \sup_{\substack{\vec x_0 \in X^M\\ \seq{u}\in U^N}}  |\decoM{\ell}(\vec x(n), \revv{\seqq{u}}(n)) - \ell(\muhat(n), \revv{\seqq{u}}(n))| \\
        & \hspace{0.5cm} + \sum_{n=0}^{N-1} \sum_{i=1}^{n-1} L_\ell L_f^{i-1} \sup_{\substack{\vec x_0 \in X^M\\ \seq{u}\in U^N}}  \|\KMEk(\muhat(n-i+1)) \\
        & \hspace{1cm} - F(\muhat(n-i), \revv{\seqq{u}}(n-i))\|_k \\
    & \rightarrow 0 \quad \text{ for } M\rightarrow \infty.
\end{align*} 
This concludes the proof.\QED

\subsection{Relaxed dynamic programming}
Relaxed dynamic programming has been used frequently in the analysis of NMPC.
We now present a mean field limit variant thereof, which can be used to derive performance bounds for mean field NMPC.
This result generalizes \rev{\cite[Proposition~1]{herty2017performance}} to a wide class of systems and feedback maps.
\begin{theorem} \label{thm:mflRDP}
Assume that $U$ is convex. Consider $\tilde V_M: X^M\rightarrow \Rnn$, $M\in\Np$, such that
1) $ \forall M\in\Np, \vec x \in X^M, \sigma\in\Permutations_M:    \: \tilde V_M(\sigma \vec x)= \tilde V_M(\vec x)$;
2) $\exists L_{\tilde V}\in\Rnn \: \forall M\in\Np, \vec x,x'\in X^M$:
    \begin{equation}
        |\tilde V_M(\vec x) - \tilde V_M(\vec x')| \leq L_{\tilde V_M} \MMDk(\muhat[\vec x], \muhat[\vec x']).
    \end{equation}
Let $\kappa_M: X^M\rightarrow U$, $M\in\Np$, such that
1) $ \forall M\in\Np, \forall \vec x \in X^M, \sigma\in\Permutations_M, u \in U:    \: \kappa_M(\sigma\vec x)= \kappa_M(\vec x)$;
2) $\exists L_\kappa\in\Rnn \: \forall M\in\Np, \vec x,x'\in X^M$:
    \begin{equation}
        \|\kappa_M(\vec x) - \kappa_M(\vec x') \|_H \leq L_\kappa \MMDk(\muhat[\vec x], \muhat[\vec x']).
    \end{equation}

Assume that there exists $\alpha\in(0,1]$ such that for all $M\in\Np$ and $\vec x \in X^M$ we have
\begin{equation}
    \tilde V_M(\vec x) \geq \tilde V_M(\decoM{f}(\vec x, \kappa_M(\vec x)) + \alpha \decoM{\ell}(\vec x, \kappa_M(\vec x)).
\end{equation}
Then there exists a strictly increasing sequence $(M_m)_m$, an $L_{\tilde V}$-Lipschitz continuous function $\tilde V: (\Pbb(X),\MMDk) \rightarrow \Rnn$ and a map $\kappa:\Pbb(X)\rightarrow H$ such that for all $\mu\in\Pbb(X)$ we have
\begin{equation} \label{eq:mflRDP}
    \tilde V(\mu) \geq \tilde V(f(\mu,\kappa(\mu))) + \alpha \ell(\mu, \kappa(\mu)).
\end{equation}
\end{theorem}
\begin{proof}
Under the given assumptions, Theorem \ref{thm:mflFbMap} is applicable to $(\kappa_M)_M$, so there exist a subsequence $(\kappa_{M_p})_p$ and an $L_\kappa$-Lipschitz continuous map $\kappa: \Pbb(X)\rightarrow U$ such that $\sup_{\vec x \in X^{M_p}} \|\kappa_{M_p}(\vec x) - \kappa(\muhat[\vec x])\|_H \rightarrow 0$ for $p\rightarrow \infty$.
Since $k$ is characteristic, we can define $\kappa=K\circ \KMEk^{-1}\lvert_{\KMEk(\Pbb(X))}$.
An inspection of the proof Proposition \ref{prop:mflFuncs} reveals that it applies to $(\tilde V_{M_p})_p$, so there exists a subsequence $(M_{p_m})_m$ and a function $\tilde V: \Pbb(X)\rightarrow\Rnn$ such that
$\tilde V_{p_m} \convmode{\Pone} \tilde V$ for $m\rightarrow\infty$,
and for all $\mu,\mu'\in \Pbb(X)$ 
we also have $ |\tilde V(\mu) - \tilde V(\mu')| \leq L_{\tilde V}\MMDk(\mu,\mu').$
To simplify the notation, we denote $(M_{p_m})_m$ by $(M_m)_m$ from now on.
Let now $\mu\in\Pbb(X)$ and $\epsilon>0$ be arbitrary. There exists $\vec x_M\in X^M$ such that $\MMDk(\muhat[\vec x_M],\mu)\rightarrow 0$.
Define $\hat \epsilon = \epsilon/5(L_{\tilde V}(1+L_f(1+L_\kappa)) + \alpha L_\ell(1+L_\kappa))$ and choose $m\in\Np$ such that
\begin{align*}
    \MMDk(\muhat[\vec x],\mu) & \leq \hat\epsilon \\
    \sup_{\vec x \in X^{M_m}} |\tilde V_{M_m}(\vec x) - \tilde V(\muhat[\vec x])| & \leq \frac{\epsilon}{10} =: \epsilon_{\tilde V} \\
    \sup_{\vec x \in X^{M_m}} \|\kappa_{M_m}(\vec x) - \kappa(\muhat[\vec x])\|_H & \leq  \epsilon/5(L_{\tilde V}L_f + \alpha L_\ell) =: \epsilon_\kappa \\
    \sup_{\substack{\vec x \in X^{M_m}\\ u \in U}} |\deco{\ell}{M_m}(\vec x, u) - \ell(\muhat[\vec x], u)| & \leq \epsilon/5\alpha =:\epsilon_\ell \\
\end{align*}
and
\begin{equation*}
    \sup_{\substack{\vec x \in X^{M_m}\\ u \in U}} \| \KMEhat(\deco{f}{M_m}(\vec x,u)) - F(\KMEhat(\vec x),u)\|_k \leq \epsilon/5L_{\tilde V} = : \epsilon_f.
\end{equation*}
We now have
\begin{align*}
& \| \kappa(\mu) - \kappa_{M_m}(\vec x_{M_m})\|_H = \| \kappa(\mu) - \kappa(\muhat[\vec x_{M_m}]) \|_H \\
    & \hspace{0.5cm} + \|\kappa(\muhat[\vec x_{M_m}]) - \kappa_{M_m}(\vec x_{M_m})\|_H \\
& \leq L_\kappa \MMDk(\muhat[\vec x_{M_m}],\mu) + \epsilon_\kappa 
\leq L_\kappa \hat\epsilon + \epsilon_\kappa
\end{align*}
and 
\begin{align*}
    & \MMDk\left(\muhat\left[\deco{f}{M_m}(\vec x_{M_m},\kappa_{M_m}(\vec x_{M_m}))\right], f(\mu,\kappa(\mu))\right) \\
    & \leq \| F(\mu,\kappa(\mu)) - F(\muhat[\vec x_{M_m}], \kappa_{M_m}(\vec x_{M_m})) \|_k \\
        & \hspace{0.5cm} + \|  F(\muhat[\vec x_{M_m}], \kappa_{M_m}(\vec x_{M_m})) \\
        & \hspace{0.5cm} - \KMEk\left(\deco{f}{M_m}(\vec x_{M_m},\kappa_{M_m}(\vec x_{M_m}))\right)\|_k \\
    & \leq L_f\left\|(\KMEk(\mu),\kappa(\mu)) - (\KMEk(\muhat[\vec x_{M_m}], \kappa_{M_m}(\vec x_{M_m})) \right\|_{H_k\times H} \\
        & \hspace{0.5cm} +  \sup_{\substack{\vec x \in X^{M_m}\\ u \in U}} \| \KMEk(\deco{f}{M_m}(\vec x,u)) - F(\KMEk(\muhat[\vec x]),u)\|_k \\
    & \leq L_f\left(\MMDk(\muhat[\vec x_{M_m}], \mu) + \|\kappa_{M_m}(\vec x_{M_m}) - \kappa(\mu)\|_H \right) + \epsilon_f \\
    & \leq L_f\hat\epsilon + L_f( L_\kappa \hat\epsilon + \epsilon_\kappa) + \epsilon_f
\end{align*}
as well as $\tilde V(\mu)\geq  \tilde V_{M_m}(\vec x_{M_m}) - \epsilon_{\tilde V} - L_{\tilde V}\hat\epsilon$.
Finally,
\begin{align*}
\tilde V(\mu) %
& \geq \tilde V_{M_m}(\vec x_{M_m}) - \epsilon_{\tilde V} - L_{\tilde V}\hat\epsilon \\
& \geq \tilde V_{M_m}(\deco{f}{M_m}(\vec x_{M_m},\kappa_{M_m}(\vec x_{M_m}))) \\
    & \hspace{0.5cm} + \alpha\deco{\ell}{M_m}(\vec x_{M_m}, \kappa_{M_m}(\vec x_{M_m}))  - \epsilon_{\tilde V} - L_{\tilde V}\hat\epsilon \\
& \geq \tilde V(\muhat[\deco{f}{M_m}(\vec x_{M_m},\kappa_{M_m}(\vec x_{M_m}))]) \\
    & \hspace{0.5cm} + \alpha\ell(\muhat[\vec x_{M_m}], \kappa_{M_m}(\vec x_{M_m})) - 2\epsilon_{\tilde V} - \alpha\epsilon_\ell - L_{\tilde V}\hat\epsilon \\
& \geq  \tilde V(f(\mu,\kappa(\mu))) + \alpha\ell(\mu,\kappa(\mu))  \\
    & \hspace{0.5cm} - L_{\tilde V}\left(L_f\hat\epsilon + L_f( L_\kappa \hat\epsilon + \epsilon_\kappa) + \epsilon_f\right) \\
    & \hspace{0.5cm} - \alpha L_\ell\left(\hat\epsilon +  L_\kappa \hat\epsilon + \epsilon_\kappa\right)
    - 2\epsilon_{\tilde V} - \alpha\epsilon_\ell - L_{\tilde V}\hat\epsilon \\
& =  \tilde V(f(\mu,\kappa(\mu))) + \alpha\ell(\mu,\kappa(\mu)) - \epsilon.
\end{align*}
Since $\epsilon>0$ was arbitrary, we get \eqref{eq:mflRDP}.\QED
\end{proof}

\section{Conclusion} \label{sec:conclusion}
We established the first MFL existence results for general deterministic discrete-time dynamical systems and corresponding feedback maps,
\rev{using KMEs as a new tool in this context.}
We applied the existence results in an optimal control setting, %
and established an MFL variant of the relaxed dynamic programming principle.
Our results are very general and henceforth applicable to a wide variety of discrete-time systems, in particular, discrete-time MAS as arising in many fields of science and engineering\rev{, e.g., robotics.}
The existence results in Section \ref{sec:mfl} provide a rigorous foundation for applying mean field approaches to such MAS, and hence form a theoretical basis for discrete-time mean field control.
Ongoing work is concerned with optimal control applications, in particular, NMPC %
\rev{for} large\revv{-}scale MAS.

\section*{Acknowledgements}
We thank P.-F. Massiani, \revv{F. Solowjow, and the reviewers for detailed and helpful comments}.

\bibliographystyle{IEEEtran} 
\bibliography{refs.bib} 

\end{document}